\documentclass{amsart}

\usepackage{amsmath,amssymb,amsfonts}
\usepackage{amsthm}
\usepackage{mathtools}
\usepackage[author-year]{amsrefs}

\usepackage{booktabs}
\usepackage{graphicx}
\usepackage{array}
\usepackage{adjustbox}
\usepackage{booktabs}

\usepackage{enumitem}
\usepackage{xurl}
\usepackage[hidelinks]{hyperref}
\usepackage{microtype}

\allowdisplaybreaks

\newtheorem{theorem}{Theorem}[section]
\newtheorem{proposition}[theorem]{Proposition}
\newtheorem{lemma}[theorem]{Lemma}
\newtheorem{corollary}[theorem]{Corollary}
\theoremstyle{definition}
\newtheorem{definition}[theorem]{Definition}
\theoremstyle{remark}

\newenvironment{remark*}{\noindent \textbf{Remark.}}{}

\makeatletter
\renewcommand{\subsubsection}{\@startsection{subsubsection}{3}%
  \z@{.5\linespacing\@plus.7\linespacing}{-.5em}%
  {\normalfont\bfseries}}
\makeatother

\title[Interior--Boundary Assortativity Profiles]{Interior--Boundary Assortativity Profiles on Networks and Applications to SIS Epidemic Dynamics}

\author{Moses~Boudourides}
\address{School of Professional Studies, Northwestern University}
\email{Moses.Boudourides@northwestern.edu}

\date{}


\begin{document}

\begin{abstract}
We introduce \emph{interior--boundary assortativity profiles} as a structural refinement of Newman’s assortativity coefficient and show that they arise naturally from \emph{epidemic dynamics} on networks.
Given a fixed partition of the node set, edges are stratified according to whether their endpoints are interior or boundary nodes relative to the partition, yielding type-restricted assortativity components.
We prove an exact decomposition theorem showing how classical scalar assortativity collapses heterogeneous interior--boundary interactions into a single number.

We then study an SIS epidemic model and consider equilibrium infection probabilities as node attributes.
Under mild connectivity and positivity assumptions, we show that \emph{boundary dominance}---a dynamical concentration of infection mass on interface nodes---implies a strictly negative boundary-to-interior assortativity component.
This establishes a rigorous link between the external weighted adjacency influx at partition boundaries, equilibrium pressure geometry, and the sign structure of assortative mixing induced by the dynamics.

Our results demonstrate that assortativity profiles encode dynamical information invisible to scalar summaries and provide a mathematically grounded bridge between network partition geometry and nonlinear dynamics on graphs.
\end{abstract}

\maketitle

\section{Introduction}

Assortative mixing is a fundamental concept in network theory, quantifying the tendency of nodes with similar attributes to connect to one another.
Since the publication of his seminal work, Newman’s assortativity coefficient has been widely used to characterize structural, functional, and social properties of networks.
However, a single scalar coefficient necessarily aggregates heterogeneous interactions and may obscure fine-scale structure, particularly in networks with pronounced interfaces between constituent groups.

At the same time, many node attributes of interest are not exogenous labels, but arise endogenously from dynamical processes on the network.
Epidemic spreading, opinion formation, synchronization, and flow-based processes all generate node-level quantities whose spatial distribution reflects both the network topology and the underlying dynamics.
Understanding how such attributes mix across network interfaces is therefore a question at the intersection of network structure and dynamical systems.

In this work we focus on networks equipped with a fixed partition and introduce a systematic refinement of assortativity based on \emph{interior--boundary stratification}.
Nodes are classified as interior or boundary depending on whether their incident edges remain within their group or cross group boundaries.
This induces a canonical decomposition of (directed) edges into interior--interior, interior--boundary, boundary--interior, and boundary--boundary types, and leads naturally to a vector-valued assortativity profile.

Our first contribution is a general \emph{decomposition theorem} showing that scalar assortativity is a weighted combination of type-restricted assortativities together with an explicit between-type mean-shift term.
This result makes precise how one-number assortativity collapses heterogeneous interaction regimes.

We then turn to dynamics.
We consider a (directed) SIS epidemic model and study the equilibrium infection probabilities as node attributes.
Under an explicit condition on the external weighted adjacency influx at partition boundaries, boundary nodes receive strictly higher mean infection pressure than interior nodes at the endemic equilibrium. Under the assumption of boundary dominance, we prove that this mechanism enforces a strictly negative boundary-to-interior assortativity component, even when global assortativity may be small or vanish entirely.

By deriving assortativity properties directly from a nonlinear dynamical system, we shift the interpretation of assortativity profiles from descriptive statistics to structural observables of network dynamics.
This perspective aligns naturally with the theory of dynamical systems on networks and provides a rigorous framework for linking partition geometry, flow bottlenecks, and emergent node attributes.

Methodologically, our analysis follows a two-step argument. First, under an explicit condition on the external weighted adjacency influx, the endemic equilibrium of the (directed) SIS dynamics exhibits boundary pressure amplification: boundary nodes receive strictly higher mean infection pressure than interior nodes. Second, under the assumption of boundary dominance, the interior--boundary assortativity framework yields a rigorous sign prediction for the boundary-to-interior assortativity component. In this way we establish a principled implication chain:
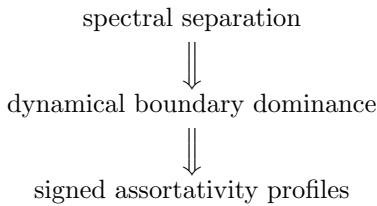
\[
\begin{array}{c}
\text{external weighted adjacency influx condition} \\
\Big\Downarrow \\
\text{boundary pressure amplification} \\
\Big\Downarrow \\
\text{(under boundary dominance assumption)} \\
\Big\Downarrow \\
\text{signed assortativity profiles}
\end{array}
\]

From the perspective of dynamical systems, assortativity profiles act as coarse observables encoding how nonlinear dynamics interact with partition geometry. The results below show that the sign structure of these profiles can be derived analytically from flow bottlenecks and equilibrium structure, placing assortative mixing within the scope of rigorous network dynamical systems theory.

Throughout this paper we develop the theory primarily for directed graphs, since the dynamical framework based on directed conductance and SIS epidemic dynamics is naturally formulated in that setting. Whenever appropriate, we indicate the corresponding simplifications for undirected graphs in separate remarks and subsections. Thus, unless explicitly stated otherwise, all definitions and results should be understood in the directed setting.

The remainder of the paper is organized as follows.
Section~2 introduces graph partitions and nodal attributes and fixes the notation used throughout.
Section~3 reviews Newman's assortativity coefficient for scalar and categorical attributes.
Section~4 introduces interior--boundary assortativity profiles by stratifying nodes and edges relative to a fixed partition.
Section~5 establishes a profile collapse theorem showing how scalar assortativity aggregates heterogeneous interaction regimes.
Section~6 develops sign mechanisms for interface-sensitive attributes and provides sufficient conditions for sign-definite interior--boundary components.
Section~7 hosts the concept of directed conductance and discusses certain directed spectral tools.
Section~8 introduces an undirected spectral proxy and discusses its relevance for directed networks via projection.
Section~9 studies SIS epidemic dynamics, proves boundary pressure amplification under an explicit external adjacency influx condition, and derives sign-definite assortativity profiles under the assumption of boundary dominance.
Section~10 concludes with a discussion of implications and possible extensions.

\section{Graph partitions and nodal attributes}

Throughout, $G=(V,E)$ (unless otherwise stated) is a directed graph (a digraph) with finite node set \(V\), where $E\subseteq V\times V$ and an arc (directed edge) $(u,v)\in E$ points from source $u$ to target $v$. In addition, we denote the adjacency matrix by $A=(A_{ij})$, where $A_{ij}\ge 0$ is the weight of the arc $(i,j)$ (and $A_{ij}\in\{0,1\}$ if unweighted).
The corresponding adjustments for the undirected case are addressed in a separate subsection.

A (\(k\)--)\emph{partition} of \(V\) is a family \(\mathcal P=\{V_1,\dots,V_k\}\) of nonempty, pairwise disjoint subsets whose union is \(V\),
equivalently \(V=\bigsqcup_{j=1}^k V_j\), where \(\bigsqcup\) denotes \emph{disjoint union}
(i.e.\ \(V=\bigcup_{j=1}^k V_j\) and \(V_i\cap V_j=\varnothing\) for \(i\neq j\)).
Let \(\mathfrak{P}_k\) be the set of all such (unlabeled) \(k\)-partitions.

A nodal \emph{attribute} (labeling) is a map \(a:V\to\mathbb{R}\), and we denote by \(\mathfrak{A}:=\{a:V\to\mathbb{R}\}\) the set of all
(real-valued) nodal attributes.
Fix \(k\ge 1\) and write \([k]=\{1,2,\dots,k\}\).
The set of \emph{enumerative} (discrete) nodal attributes with \(k\) levels is
\(\mathfrak{A}_k:=\{a:V\to [k]\}\), and \(\mathfrak{A}_k^{\mathrm{surj}}:=\{a\in\mathfrak{A}_k:\ a \text{ is surjective}\}\).

Any surjective enumerative attribute \(a\in\mathfrak{A}_k^{\mathrm{surj}}\) induces a partition via its fibers:
\(\Phi(a):=\{a^{-1}(1),\dots,a^{-1}(k)\}\in\mathfrak{P}_k\).
Conversely, given \(\mathcal P=\{V_1,\dots,V_k\}\in\mathfrak{P}_k\), choosing an ordering of its blocks defines an enumerative attribute
\(a_{\mathcal P}:V\to [k]\) by \(a_{\mathcal P}(v)=j\) iff \(v\in V_j\); a different ordering produces a labeling obtained by permuting labels.

Formally, the symmetric group \(\mathfrak{S}_k\) acts on \(\mathfrak{A}_k^{\mathrm{surj}}\) by relabeling, \((\sigma\cdot a)(v):=\sigma(a(v))\).
Then for \(a,b\in\mathfrak{A}_k^{\mathrm{surj}}\) one has \(\Phi(a)=\Phi(b)\) if and only if \(b=\sigma\cdot a\) for some \(\sigma\in\mathfrak{S}_k\).
Equivalently, \(\Phi\) induces a bijection \(\overline{\Phi}:\mathfrak{A}_k^{\mathrm{surj}}/\mathfrak{S}_k \to \mathfrak{P}_k\).

In particular, specifying a \(k\)-partition of \(V\) is the same as specifying a surjective discrete nodal attribute with values in \([k]\),
\emph{up to relabeling of the categories}: partitions correspond exactly to equivalence classes of \(k\)-labelings under permutations of \([k]\).

\section{Newman's assortativity coefficient}
\label{sec:directed_prelims}

Let $G=(V,E)$ be a graph and $x:V\to\mathbb{R}$ be a nodal attribute.

\begin{definition}[Newman’s assortativity coefficient (Pearson form) {\cites{Newman2003Mixing,Newman2010Networks}}]
Assume that each arc $e=(u,v)\in E$ has a nonnegative weight $w(e)$
(with $w(e)=1$ in the unweighted case), and define
\[
m:=\sum_{e\in E} w(e).
\]

Equip the finite set $E$ with the probability measure
\[
\mathbb{P}(e)=\frac{w(e)}{m}, \qquad e\in E.
\]
Let $(U,V)$ be a random directed edge (arc) with distribution $\mathbb{P}$ and, 
from the nodal attribute $x$, define two real-valued random variables
\[
X:=x(U), \qquad Y:=x(V).
\]

Their expectations are
\[
\mathbb{E}[X]
=\sum_{(u,v)\in E} x(u)\,\frac{w(u,v)}{m},\qquad
\mathbb{E}[Y]
=\sum_{(u,v)\in E} x(v)\,\frac{w(u,v)}{m},
\]
\[
\mathbb{E}[XY]
=\sum_{(u,v)\in E} x(u)x(v)\,\frac{w(u,v)}{m},
\]
\[
\mathbb{E}[X^2]
=\sum_{(u,v)\in E} x(u)^2\,\frac{w(u,v)}{m},\qquad
\mathbb{E}[Y^2]
=\sum_{(u,v)\in E} x(v)^2\,\frac{w(u,v)}{m}.
\]
Define the variances
\[
\mathrm{Var}(X):=\mathbb{E}[X^2]-\mathbb{E}[X]^2,
\qquad
\mathrm{Var}(Y):=\mathbb{E}[Y^2]-\mathbb{E}[Y]^2.
\]

The \emph{assortativity coefficient} of the attribute $x$ is defined as the Pearson
correlation between attribute values at the tail and head of a randomly chosen
directed edge,
\[
\rho(x)
:=
\frac{\mathbb{E}[XY]-\mathbb{E}[X]\mathbb{E}[Y]}
{\sqrt{\mathrm{Var}(X)}\sqrt{\mathrm{Var}(Y)}} \in[-1,1],
\]
provided $\mathrm{Var}(X)>0$ and $\mathrm{Var}(Y)>0$.
\end{definition}

\begin{definition}[Newman’s assortativity coefficient (adjacency matrix form) {\cites{Newman2003Mixing,Newman2010Networks}}]
Suppose that an adjacency matrix $A$ is given and, identifying arc weights as $w(e) = w(i,j) = A_{i,j}$, 
let
\[
m=\sum_{i,j} A_{ij},\qquad
k_i^{\mathrm{out}}=\sum_j A_{ij},\qquad
k_j^{\mathrm{in}}=\sum_i A_{ij},
\]
and write $x_i=x(i)$. Define edge-weighted means and variances
\[
\mu_{\mathrm{out}}=\frac{1}{m}\sum_i k_i^{\mathrm{out}}x_i,
\qquad
\mu_{\mathrm{in}}=\frac{1}{m}\sum_j k_j^{\mathrm{in}}x_j,
\]
\[
\sigma_{\mathrm{out}}^2=\frac{1}{m}\sum_i k_i^{\mathrm{out}}x_i^2-\mu_{\mathrm{out}}^2,
\qquad
\sigma_{\mathrm{in}}^2=\frac{1}{m}\sum_j k_j^{\mathrm{in}}x_j^2-\mu_{\mathrm{in}}^2.
\]
Then the Pearson form of the \emph{assortativity coefficient} of the attribute $x$ is written as:
\[
\rho(x)
=
\frac{\frac{1}{m}\sum_{i,j}A_{ij}x_i x_j-\mu_{\mathrm{out}}\mu_{\mathrm{in}}}{\sigma_{\mathrm{out}}\sigma_{\mathrm{in}}}
\in[-1,1].
\]
\end{definition}

\begin{definition}[Newman's assortativity coefficient (mixing-matrix form)~\cite{Newman2003Mixing}]
\label{def:newman_mixing}
Let $x:V\to\{1,\dots,K\}$ be a discrete (categorical) attribute. We use $\mathbf{1}\{\cdot\}$ for an indicator function. Define the mixing matrix $e=(e_{pq})$ of the network with respect to $x$ by
\begin{equation*}
\begin{aligned}
e_{pq}&=\frac{1}{m}\sum_{i,j}A_{ij}\,\mathbf{1}\{x_i=p,\ x_j=q\},\\
a_p&=\sum_q e_{pq},\quad
b_q=\sum_p e_{pq}.
\end{aligned}
\end{equation*}
The mixing-matrix form of the assortativity coefficient of the discrete attribute $x$ is
\[
\rho(x)=\frac{\mathrm{Tr}(e)-\sum_{p=1}^K a_p b_p}{1-\sum_{p=1}^K a_p b_p}
\in[-1,1].
\]
\end{definition}

The mixing-matrix formulation applies only to discrete (categorical) nodal attributes; scalar attributes are handled via the Pearson and adjacency-matrix forms.

\begin{proposition}[Equivalence of assortativity formulations {\cites{Newman2003Mixing,Newman2010Networks}}]
For any graph, the Pearson form and adjacency-matrix form of assortativity
coincide for any real-valued nodal attribute.
When the attribute is discrete, these are equivalent to the mixing-matrix form.
\end{proposition}

\begin{definition}[Leicht--Newman directed modularity~\cite{LeichtNewman2008DirectedCommunities}]
Given a directed graph $G=(V,E)$ and a discrete nodal attribute $x:V\to\{1,\dots,K\}$, 
the Leicht--Newman directed modularity of a partition
$\mathcal{P}$ (with partition assignment corresponding to $x$) is
\begin{equation}
\label{eq:directed_modularity}
Q(G,\mathcal{P})
=\frac{1}{m}\sum_{i,j}\Big(A_{ij}-\gamma\,\frac{k_i^{\mathrm{out}}k_j^{\mathrm{in}}}{m}\Big)\,\mathbf{1}\{x(i)=x(j)\},
\end{equation}
where $\gamma>0$ is a resolution parameter. Moreover, we call $(G,\mathcal{P})$ \emph{$Q$-modular} if $Q(G,\mathcal{P})\ge q_0$, for a stated threshold $q_0 \geq 0$.
\end{definition}

\begin{proposition}[Assortativity and directed modularity {\cites{Newman2003Mixing,LeichtNewman2008DirectedCommunities}}]
\label{prop:assortativity_modularity}
Let $G=(V,E)$ be a directed graph and let $x:V\to\{1,\dots,K\}$ be a discrete nodal attribute
with induced partition $\mathcal P_x=\{x^{-1}(1),\dots,x^{-1}(K)\}$.
Then the mixing-matrix assortativity coefficient $\rho(x)$ satisfies
\[
\rho(x)
=\frac{Q(G,\mathcal P_x)}{1-\sum_{p=1}^K a_p b_p},
\]
provided that $1-\sum_{p=1}^K a_p b_p>0$,
where $Q(G,\mathcal P_x)$ is the Leicht--Newman directed modularity
\eqref{eq:directed_modularity} with $\gamma=1$ (the standard directed modularity), and
$a_p,b_p$ are the row and column marginals of the directed mixing matrix.

Thus, for every fixed partition $\mathcal P_x$, the discrete assortativity coefficient
is a normalized version of the corresponding directed modularity, with normalization
factor $1-\sum_{p=1}^K a_p b_p$.
\end{proposition}

The interplay of stratifying and attribute-induced partitions exhibits several interesting features. So, let us fix a partition $\mathcal P=\{V_1,\dots,V_k\}$ (the \emph{stratifying partition}) and let $x:V\to\{1,\dots,L\}$ be a discrete nodal attribute.
Let $\mathcal Q=\{H_1,\dots,H_L\}$ be the partition induced by $x$, where $H_r:=\{v\in V:\ x(v)=r\}$.
Node-overlap summaries of $\mathcal P$ and $\mathcal Q$ (e.g.\ Jaccard indices
$J(V_k,H_r)=|V_k\cap H_r|/|V_k\cup H_r|$) quantify agreement at the level of node sets, but the assortativity profiles depend on how \emph{edge weight}
is distributed across the refinement blocks $V_k\cap H_r$ (and their interior/boundary splits).

\begin{proposition}[Two-partition refinement representation of the profile]
\label{thm:two_partitions_refinement}
Let $A=(A_{uv})$ be the (possibly weighted) adjacency matrix and fix $\mathcal P$.
For each $k$, write $I_k:=\mathrm{Int}(V_k)$ and $B_k:=\mathrm{Bdy}(V_k)$, and form the four $\mathcal P$-strata of arcs
\[
E \;=\; E_{I\to I}\ \bigsqcup\ E_{I\to B}\ \bigsqcup\ E_{B\to I}\ \bigsqcup\ E_{B\to B},
\]
with $E_{I\to I},E_{I\to B},E_{B\to I},E_{B\to B}$ defined as above. 
Fix a type $T\in\{I\!\to\!I,\ I\!\to\!B,\ B\!\to\!I,\ B\!\to\!B\}$ and define the type-restricted adjacency
$A^{(T)}_{uv}:=A_{uv}\mathbf 1\{(u,v)\in E_T\}$ and total mass $m_T:=\sum_{u,v}A^{(T)}_{uv}$.
For attribute labels $(r,s)\in\{1,\dots,L\}^2$, define the type--label masses
\[
m_T^{rs}:=\sum_{u,v} A^{(T)}_{uv}\,\mathbf 1\{u\in H_r,\ v\in H_s\}.
\]
If $m_T>0$, the type-restricted mixing matrix is $e^{(T)}_{rs}:=m_T^{rs}/m_T$ and the corresponding categorical coefficient is
\[
\rho_T(x;\mathcal P)
=\frac{\mathrm{Tr}(e^{(T)})-\sum_{r=1}^L a_r^{(T)} b_r^{(T)}}{1-\sum_{r=1}^L a_r^{(T)} b_r^{(T)}},
\qquad
a_r^{(T)}:=\sum_s e^{(T)}_{rs},\ \ b_s^{(T)}:=\sum_r e^{(T)}_{rs},
\]
whenever the denominator is nonzero. The interior--boundary assortativity profile is then
\[
\boldsymbol{\rho}(x;\mathcal P)
=\big(\rho_{I\to I}(x;\mathcal P),\,\rho_{I\to B}(x;\mathcal P),\,\rho_{B\to I}(x;\mathcal P),\,\rho_{B\to B}(x;\mathcal P)\big)\in[-1,1]^4,
\]
with the convention that a component is undefined when $m_T=0$ (and/or the denominator vanishes).

Moreover, each $m_T^{rs}$ is determined by edge weights between the \emph{refinement blocks} $I_k\cap H_r$ and $B_k\cap H_r$:
\begin{align*}
m_{I\to I}^{rs}&=\sum_{k=1}^K \ \sum_{u\in I_k\cap H_r}\ \sum_{v\in I_k\cap H_s} A_{uv},\\
m_{I\to B}^{rs}&=\sum_{k=1}^K \ \sum_{u\in I_k\cap H_r}\ \sum_{v\in B_k\cap H_s} A_{uv},\\
m_{B\to I}^{rs}&=\sum_{k=1}^K \ \sum_{u\in B_k\cap H_r}\ \sum_{v\in I_k\cap H_s} A_{uv},\\
m_{B\to B}^{rs}&=\sum_{p=1}^K\sum_{q=1}^K \ \sum_{u\in B_p\cap H_r}\ \sum_{v\in B_q\cap H_s} A_{uv}.
\end{align*}
In particular, the profile $\boldsymbol{\rho}(x;\mathcal P)$ is not determined by node-overlap indices (such as Jaccard similarities between
blocks of $\mathcal P$ and $\mathcal Q$) alone, but by the edge-weighted interaction pattern among the refinement blocks.
\end{proposition}

\begin{remark*}
The proposition, which admits a straightforward proof, shows that assortativity profiles capture interactions
between partitions at the level of edge flow, not merely node overlap.
\end{remark*}

\subsection*{Adjustments in the undirected case}
\label{subsec:undirected_adjustments}

All the previous discussion was referring to directed graphs by default. The undirected case is obtained by the following specializations. Let $G$ be an undirected (possibly weighted) graph with symmetric adjacency matrix $A=(A_{ij})$, where $A_{ij}\ge 0$ is the edge weight
(and $A_{ij}\in\{0,1\}$ if unweighted). Set
\[
m=\frac{1}{2}\sum_{i,j}A_{ij},\qquad k_i=\sum_j A_{ij}.
\]
Then, considering a scalar nodal attribute $x$, the covariance/modularity-matrix identity
\[
\mathrm{cov}(x_i,x_j)=\frac{1}{2m}\sum_{i,j}\Big(A_{ij}-\frac{k_i k_j}{2m}\Big)x_i x_j
\]
reduces the attribute assortativity coefficient (Newman, Eq.~(7.80)) to the form
\[
\rho(x)=\frac{\sum_{i,j}\big(A_{ij}-\frac{k_i k_j}{2m}\big)x_i x_j}{\sum_{i,j}\big(k_i\delta_{ij}-\frac{k_i k_j}{2m}\big)x_i x_j},
\]
where $\delta_{ij}$ is the Kronecker delta: $\delta_{ij}=1$ if $i=j$ and $\delta_{ij}=0$ otherwise.

Furthermore, the Leicht--Newman directed modularity for discrete nodal attributes  reduces to the standard Newman--Girvan modularity~\cite{NewmanGirvan2004}
\[
Q(G,\mathcal{P})=\frac{1}{2m}\sum_{i,j}\Big(A_{ij}-\gamma\,\frac{k_i k_j}{2m}\Big)\,\mathbf{1}\{x(i)=x(j)\}.
\]

\section{Interior--boundary assortativity profiles}
\label{sec:profiles_three}

Let $G = (V,E)$ be a directed graph.
For $v\in V$, define in- and out-neighborhoods
$N^{\mathrm{in}}(v)=\{u\in V:(u,v)\in E\}$ and
$N^{\mathrm{out}}(v)=\{w\in V:(v,w)\in E\}$.
For each group $V_k\in\mathcal{P}$, define its \emph{interior} and \emph{boundary}:
\begin{equation}
\label{eq:I_Bd}
\begin{aligned}
\mathrm{Int}(V_k)
&=\{v\in V_k:\ N^{\mathrm{in}}(v)\cup N^{\mathrm{out}}(v)\subseteq V_k\},\\
\mathrm{Bdy}(V_k)
&=V_k\setminus \mathrm{Int}(V_k).
\end{aligned}
\end{equation}
Thus, boundary nodes are precisely those having at least one incident arc whose other endpoint lies in a different group.

\medskip
\noindent\textbf{Arc stratification.}
Every arc $(u,v)\in E$ has a tail $u$ and a head $v$, each of which is either interior or boundary with respect to its unique group in $\mathcal P$.
This induces a canonical disjoint decomposition of the entire arc set:
\[
E
=\;E_{I\to I}\ \bigsqcup\ E_{I\to B}\ \bigsqcup\ E_{B\to I}\ \bigsqcup\ E_{B\to B},
\]
where
\begin{align*}
E_{I\to I}
&:=\{(u,v)\in E:\ \exists k,\ u\in\mathrm{Int}(V_k),\ v\in\mathrm{Int}(V_k)\},\\
E_{I\to B}
&:=\{(u,v)\in E:\ \exists k,\ u\in\mathrm{Int}(V_k),\ v\in\mathrm{Bdy}(V_k)\},\\
E_{B\to I}
&:=\{(u,v)\in E:\ \exists k,\ u\in\mathrm{Bdy}(V_k),\ v\in\mathrm{Int}(V_k)\},\\
E_{B\to B}
&:=\{(u,v)\in E:\ \exists p,q,\ u\in\mathrm{Bdy}(V_p),\ v\in\mathrm{Bdy}(V_q)\}.
\end{align*}
The first three strata consist necessarily of intra-group arcs, whereas $E_{B\to B}$ contains both
within-group and cross-group boundary--boundary arcs.
In particular, if the graph is multipartite with respect to $\mathcal P$, then
$\mathrm{Int}(V_k)=\varnothing$ for all $k$ and $E_{B\to B}=E$.

For each type $T\in\{I\!\to\!I,\ I\!\to\!B,\ B\!\to\!I,\ B\!\to\!B\}$, define the corresponding type-restricted adjacency matrix
\begin{equation*}
\begin{aligned}
A^{(T)}_{ij}&=A_{ij}\,\mathbf{1}\{(i,j)\in E_T\},\\
m_T&=\sum_{i,j}A^{(T)}_{ij},\qquad
k^{\mathrm{out},T}_i=\sum_j A^{(T)}_{ij},\qquad
k^{\mathrm{in},T}_j=\sum_i A^{(T)}_{ij}.
\end{aligned}
\end{equation*}

Thus, 
\begin{equation*}
\begin{aligned}
A&=\sum_T A^{(T)} \\
m&=\sum_T m_T.
\end{aligned}
\end{equation*}

\subsubsection*{Profiles for scalar attributes}
Let $x:V\to\mathbb{R}$ and write $x_i=x(i)$. For each type $T$ with $m_T>0$, define
\begin{equation*}
\begin{aligned}
\mu_{\mathrm{out},T}&=\frac{1}{m_T}\sum_i k_i^{\mathrm{out},T}x_i,
\quad
\mu_{\mathrm{in},T}=\frac{1}{m_T}\sum_j k_j^{\mathrm{in},T}x_j,\\
\sigma_{\mathrm{out},T}^2
&=\frac{1}{m_T}\sum_i k_i^{\mathrm{out},T}x_i^2-\mu_{\mathrm{out},T}^2,\\
\sigma_{\mathrm{in},T}^2
&=\frac{1}{m_T}\sum_j k_j^{\mathrm{in},T}x_j^2-\mu_{\mathrm{in},T}^2.
\end{aligned}
\end{equation*}
The type-restricted assortativity (in Newman’s Pearson or adjacency form) is
\begin{equation*}
\begin{aligned}
\rho_T(x;\mathcal{P})
&=
\frac{\frac{1}{m_T}\sum_{i,j}A^{(T)}_{ij}x_i x_j-\mu_{\mathrm{out},T}\mu_{\mathrm{in},T}}
{\sigma_{\mathrm{out},T}\sigma_{\mathrm{in},T}} \\
&=
\frac{\frac{1}{m_T}\sum_{i,j}\Big(A^{(T)}_{ij}
-\frac{k_i^{\mathrm{out},T}k_j^{\mathrm{in},T}}{m_T}\Big)x_i x_j}
{\sigma_{\mathrm{out},T}\sigma_{\mathrm{in},T}}
\in[-1,1].
\end{aligned}
\end{equation*}
The \emph{interior--boundary assortativity profile} is
\[
\boldsymbol{\rho}(x;\mathcal{P})
=
\big(\rho_{I\to I}(x;\mathcal{P}),\,
\rho_{I\to B}(x;\mathcal{P}),\,
\rho_{B\to I}(x;\mathcal{P}),\,
\rho_{B\to B}(x;\mathcal{P})\big)
\in[-1,1]^4,
\]
with the convention that components are undefined whenever $m_T=0$. Thus, each component measures assortative mixing restricted to a specific interior–boundary interaction regime.

\subsubsection*{Profiles for categorical attributes}
Let $x:V\to\{1,\dots,K\}$ be a discrete (categorical) nodal attribute.
For each type $T$ with $m_T>0$, define the type-restricted mixing matrix
$e^{(T)}=(e^{(T)}_{pq})$ by
\begin{equation*}
\begin{aligned}
e^{(T)}_{pq}
&=\frac{1}{m_T}\sum_{i,j}A^{(T)}_{ij}\,
\mathbf{1}\{x(i)=p,\ x(j)=q\},\\
a^{(T)}_p&=\sum_q e^{(T)}_{pq},
\qquad
b^{(T)}_q=\sum_p e^{(T)}_{pq}.
\end{aligned}
\end{equation*}
The type-restricted assortativity coefficient is defined as
\[
\rho_T(x;\mathcal{P})
=
\frac{\mathrm{Tr}(e^{(T)})-\sum_{p=1}^K a^{(T)}_p b^{(T)}_p}
{1-\sum_{p=1}^K a^{(T)}_p b^{(T)}_p}
\in[-1,1],
\]
provided that $1-\sum_{p=1}^K a^{(T)}_p b^{(T)}_p>0$.

The \emph{interior--boundary assortativity profile} for a categorical attribute is
\[
\boldsymbol{\rho}(x;\mathcal{P})
=
\big(\rho_{I\to I}(x;\mathcal{P}),\,
\rho_{I\to B}(x;\mathcal{P}),\,
\rho_{B\to I}(x;\mathcal{P}),\,
\rho_{B\to B}(x;\mathcal{P})\big)
\in[-1,1]^4,
\]
with undefined components when $m_T=0$.

\subsection*{Adjustments in the undirected case}
\label{subsec:undirected_adjustments_p}
If $G$ is undirected, replace $N^{\mathrm{in}}(v)\cup N^{\mathrm{out}}(v)$ by the undirected neighbor set $N(v)$ in the definition of the interior:
\begin{equation*}
\begin{aligned}
\mathrm{Int}(V_k)&=\{v\in V_k:\ N(v)\subseteq V_k\},\\
\mathrm{Bdy}(V_k)&=V_k\setminus \mathrm{Int}(V_k).
\end{aligned}
\end{equation*}
The (undirected) edge set decomposes into three strata:
\[
E \;=\; E_{II}\ \bigsqcup\ E_{IB}\ \bigsqcup\ E_{BB},
\]
where $E_{II}$ (resp.\ $E_{IB}$) consists of edges with both endpoints interior (resp.\ one interior and one boundary) of the \emph{same} group,
and $E_{BB}$ consists of edges with both endpoints boundary (including both intra-group and cross-group boundary--boundary edges).
Accordingly, the profile dimension is 3 and one writes
\[
\boldsymbol{\rho}(x;\mathcal{P})=\big(\rho_{II}(x;\mathcal{P}),\,\rho_{IB}(x;\mathcal{P}),\,\rho_{BB}(x;\mathcal{P})\big)\in[-1,1]^3,
\]
with the same convention about undefined components when the corresponding stratum is empty (and/or variances vanish).

\subsection*{Multipartite undirected graphs with consistent attributes}
\label{subsec:multipartite_consistent}

Let $G=(V,E)$ be an undirected (possibly weighted) graph and let $\mathcal P=\{V_1,\dots,V_k\}$ be a partition of $V$.
A discrete nodal attribute $x:V\to [K]$ is called \emph{$\mathcal P$-consistent} if it induces $\mathcal P$, i.e.
\[
V_k = x^{-1}(k)\qquad (k=1,\dots,K),
\]
equivalently, $x$ is constant on each $V_k$ with distinct values across the $V_k$'s (up to relabeling).

We say that $G$ is \emph{$\mathcal P$-multipartite} if there are no within-group edges, i.e.
\[
\{u,v\}\in E \ \Longrightarrow\ \exists\,p\neq q \text{ such that } u\in V_p,\ v\in V_q,
\]
or equivalently $E\cap \binom{V_k}{2}=\varnothing$ for all $k$.

\begin{proposition}[Profiles for $\mathcal P$-multipartite graphs]
\label{prop:multipartite_profile}
Assume that $G$ is $\mathcal P$-multipartite and that $x:V\to[K]$ is $\mathcal P$-consistent.
Then every non-isolated node is a boundary node (relative to $\mathcal P$), and isolated nodes (if any) are interior.
Consequently, the undirected strata satisfy
\[
E_{II}=\varnothing,\qquad E_{IB}=\varnothing,\qquad E_{BB}=E,
\]
and the interior--boundary assortativity profile collapses to its $BB$ component:
\[
\boldsymbol{\rho}(x;\mathcal P)=\big(\rho_{II}(x;\mathcal P),\,\rho_{IB}(x;\mathcal P),\,\rho_{BB}(x;\mathcal P)\big)
=\big(\text{undef},\,\text{undef},\,\rho_{BB}(x;\mathcal P)\big),
\]
with
\[
\rho_{BB}(x;\mathcal P)=\rho(x),
\]
i.e.\ the $BB$ component coincides with Newman’s assortativity of $x$ computed on the full edge set.

Moreover, letting $m=\frac12\sum_{i,j}A_{ij}$ and $k_i=\sum_j A_{ij}$, define the edge-end fractions
\[
a_p := \frac{1}{2m}\sum_{i\in V_p} k_i \qquad (p=1,\dots,K),
\]
so that $\sum_{p=1}^K a_p=1$. In the $\mathcal P$-multipartite case the mixing matrix has zero diagonal, hence
\[
\rho(x)= -\frac{\sum_{p=1}^K a_p^2}{1-\sum_{p=1}^K a_p^2}\le 0,
\]
with equality $\rho(x)=-1$ if and only if the edge mass is supported on exactly two parts with $a_p=a_q=\tfrac12$
(i.e.\ the edge-induced subgraph is effectively bipartite).
Hence, multipartite structure enforces non-positive assortativity for $\mathcal P$-consistent attributes.
\end{proposition}

\subsection*{Participation coefficient as a node-level interface descriptor}
\label{subsec:participation}

The interior--boundary framework and the assortativity profiles introduced above are edge-level, second-order summaries of attribute mixing across interfaces.
It is useful to complement them with a standard node-level descriptor that quantifies how broadly a node’s connections are distributed across groups.
For this purpose we use the \emph{participation coefficient}
introduced by Guimer\`a and Amaral~\cite{GuimeraAmaral2005}.

\begin{definition}[Participation coefficient {\cite{GuimeraAmaral2005}}]
Let $G=(V,E)$ be a directed (possibly weighted) graph with partition
$\mathcal P=\{V_1,\dots,V_K\}$.
For a node $v\in V$, define the out- and in-strengths
\[
k_v^{\mathrm{out}}=\sum_j A_{vj}, \qquad
k_v^{\mathrm{in}}=\sum_i A_{iv},
\]
and the group-resolved strengths
\[
k_{v\to V_k}^{\mathrm{out}}=\sum_{j\in V_k}A_{vj},
\qquad
k_{V_k\to v}^{\mathrm{in}}=\sum_{i\in V_k}A_{iv}.
\]
When the corresponding total strength is nonzero, the directed participation coefficients are defined as
\[
P^{\mathrm{out}}(v)
=1-\sum_{k=1}^K
\left(\frac{k_{v\to V_k}^{\mathrm{out}}}{k_v^{\mathrm{out}}}\right)^2,
\qquad
P^{\mathrm{in}}(v)
=1-\sum_{k=1}^K
\left(\frac{k_{V_k\to v}^{\mathrm{in}}}{k_v^{\mathrm{in}}}\right)^2.
\]
\end{definition}

These quantities take values in $[0,1]$ and measure, respectively, how evenly $v$’s outgoing or incoming edge weight is distributed across the groups of $\mathcal P$. Low participation indicates that most connections remain within a single group, whereas high participation indicates broad cross-group connectivity.

Participation coefficients are \emph{node-level} descriptors and do not involve attribute values.
In contrast, the interior--boundary assortativity profiles are \emph{edge-level} statistics that measure how attribute values mix across specific interface regimes ($I\!\to\!I$, $I\!\to\!B$, $B\!\to\!I$, $B\!\to\!B$).
High participation implies that a node is necessarily a boundary node, but the converse need not hold.
Moreover, participation alone does not determine the sign or magnitude of any component $\rho_T(x;\mathcal P)$.

In applications, participation coefficients provide a convenient way to quantify \emph{interface exposure} at the node level and to control for purely structural effects.
For example, negative $B\!\to\!I$ assortativity for mediation-like attributes is most informative when it persists after conditioning on boundary nodes with high $P^{\mathrm{out}}$, indicating that the effect is not driven solely by the number of groups a node connects to, but by systematic attribute contrasts across the interface.

\section{A profile collapse theorem} 
\label{sec:decomposition}
Let $a:V\to\mathbb{R}$ be a scalar nodal attribute and fix a partition $\mathcal P$. For each arc type $T$, let
\[
S_T(x)=\{(a(u),a(v)):\ (u,v)\in E_T\}
\]
denote the paired endpoint sample restricted to arcs of type $T$. We also consider the aggregated paired sample over all intra-group arcs,
\[
S_{\mathrm{in}}(a)=\bigcup_{T\in\{I\!\to\!I,I\!\to\!B,B\!\to\!I\}} S_T(a),
\]
with correlation $r_{\mathrm{in}}(a)=\mathrm{Corr}(S_{\mathrm{in}}(a))$, whenever defined.

The next theorem makes precise how a single assortativity coefficient
collapses the interior--boundary assortativity profile.

\begin{theorem}[Profile collapse theorem for scalar assortativity]
\label{thm:decomposition}
Assume all variances needed below are positive. 
Let $T$ range over $\{I\!\to\!I,I\!\to\!B,B\!\to\!I,B\!\to\!B\}$ (or $\{II,IB,BB\}$ in the undirected case). For intra-group aggregation, define
\[
\pi_T := \frac{|E_T|}{\sum_{T'\in\mathcal T_{\mathrm{in}}} |E_{T'}|},
\qquad
\mathcal T_{\mathrm{in}}=\{I\!\to\!I,I\!\to\!B,B\!\to\!I\}
\]
(or $\{II,IB\}$ in the undirected case).
Then
\begin{equation}
\label{eq:decomp}
\mathrm{Cov}\!\left(S_{\mathrm{in}}(a)\right)=\sum_T \pi_T\,\mathrm{Cov}\!\left(S_T(a)\right)\;+\;\mathrm{Cov}_{\mathrm{between}},
\end{equation}
where $\mathrm{Cov}_{\mathrm{between}}$ is an explicit between-type covariance term determined by differences of type-specific endpoint means. Consequently,
\begin{equation}
\label{eq:decomp_corr}
r_{\mathrm{in}}(a)=
\frac{\sum_T \pi_T\,\sigma_{x,T}\sigma_{y,T}\,r_T(a)\;+\;\mathrm{Cov}_{\mathrm{between}}}{\sigma_{x,\mathrm{in}}\sigma_{y,\mathrm{in}}},
\end{equation}
where $\sigma_{x,T},\sigma_{y,T}$ are within-type endpoint standard deviations and $\sigma_{x,\mathrm{in}},\sigma_{y,\mathrm{in}}$ are the endpoint standard deviations on $S_{\mathrm{in}}(a)$.
Equation~\eqref{eq:decomp_corr} shows that the scalar assortativity $r_{\mathrm{in}}(a)$ is a weighted aggregation of the profile components $\{\rho_T(a;\mathcal P)\}$, plus a between-type mean-shift term.
\end{theorem}

\begin{remark*}
Even if $r_{B\to I}(a)$ (or $r_{IB}(a)$) is strongly negative, the scalar $r_{\mathrm{in}}(a)$ may be close to zero due to (i) mixing of opposite-signed within-type correlations and (ii) cancellation with the between-type mean-shift term. This is the rigorous sense in which one-number assortativity can hide interface structure.
\end{remark*}

\begin{corollary}[Loss of resolution under scalar aggregation]
Two attributes $a,b:V\to\mathbb R$ may satisfy
$r_{\mathrm{in}}(a)=r_{\mathrm{in}}(b)$ while having distinct
interior--boundary assortativity profiles
$\boldsymbol{\rho}(a;\mathcal P)\neq \boldsymbol{\rho}(b;\mathcal P)$.
\end{corollary}

\begin{proof}[Proof of Theorem~\ref{thm:decomposition}]
The proof follows from the law of total covariance applied to the arc-type stratification.
We give the proof in a directed-first way; the undirected case is identical after replacing arcs by edges and using the
three strata $II,IB,BB$ instead of the four directed strata.

Let $E_{\mathrm{in}}$ denote the set of intra-group arcs (relative to the fixed partition), and let
\[
E_{\mathrm{in}}=\bigsqcup_{T\in\mathcal T_{\mathrm{in}}} E_T
\]
be its disjoint decomposition into the interior--boundary strata that are intra-group
(e.g.\ for directed graphs $\mathcal T_{\mathrm{in}}=\{I\!\to\!I,\ I\!\to\!B,\ B\!\to\!I\}$).
Let $w(u,v)\ge 0$ be the arc weights and write
\[
m := \sum_{(u,v)\in E_{\mathrm{in}}} w(u,v), \qquad
m_T := \sum_{(u,v)\in E_T} w(u,v),
\]
and set $\pi_T := m_T/m$.

For each arc $(u,v)\in E_{\mathrm{in}}$, define the endpoint values
\[
X_{uv}:=x(u),\qquad Y_{uv}:=x(v).
\]
Let the global endpoint means be
\[
\mu_X:=\frac{1}{m}\sum_{(u,v)\in E_{\mathrm{in}}} X_{uv},\qquad
\mu_Y:=\frac{1}{m}\sum_{(u,v)\in E_{\mathrm{in}}} Y_{uv}.
\]
For each stratum $T$, define the type-specific means
\[
\mu_{X,T}:=\frac{1}{m_T}\sum_{(u,v)\in E_T} X_{uv},\qquad
\mu_{Y,T}:=\frac{1}{m_T}\sum_{(u,v)\in E_T} Y_{uv},
\]
and type-specific covariance
\[
\mathrm{Cov}_T:=\frac{1}{m_T}\sum_{(u,v)\in E_T}\big(X_{uv}-\mu_{X,T}\big)\big(Y_{uv}-\mu_{Y,T}\big).
\]
Finally, define the global covariance on $E_{\mathrm{in}}$ by
\[
\mathrm{Cov}_{\mathrm{in}}
:=\frac{1}{m}\sum_{(u,v)\in E_{\mathrm{in}}}\big(X_{uv}-\mu_X\big)\big(Y_{uv}-\mu_Y\big).
\]

\medskip
\noindent\textit{Step 1: Express global means as mixtures of type means.}
Since $E_{\mathrm{in}}$ is a disjoint union of the $E_T$,
\begin{equation*}
\begin{aligned}
\mu_X&=\frac{1}{m}\sum_T\sum_{(u,v)\in E_T}X_{uv}\\
      &=\sum_T\frac{m_T}{m}\cdot\frac{1}{m_T}\sum_{(u,v)\in E_T}X_{uv}\\
      &=\sum_T \pi_T\,\mu_{X,T}.
\end{aligned}
\end{equation*}
Similarly, $\mu_Y=\sum_T \pi_T\,\mu_{Y,T}$.

\medskip
\noindent\textit{Step 2: Decompose the global covariance (finite-sample law of total covariance).}
Fix a stratum $T$ and write, for $(u,v)\in E_T$,
\begin{equation*}
\begin{aligned}
X_{uv}-\mu_X &= (X_{uv}-\mu_{X,T})+(\mu_{X,T}-\mu_X),\\
Y_{uv}-\mu_Y &= (Y_{uv}-\mu_{Y,T})+(\mu_{Y,T}-\mu_Y).
\end{aligned}
\end{equation*}
Multiplying and summing over $(u,v)\in E_T$ gives
\begin{align*}
\sum_{(u,v)\in E_T}(X_{uv}&-\mu_X)(Y_{uv}-\mu_Y) =\\
&=\sum_{(u,v)\in E_T}(X_{uv}-\mu_{X,T})(Y_{uv}-\mu_{Y,T}) \\
&\quad +(\mu_{Y,T}-\mu_Y)\sum_{(u,v)\in E_T}(X_{uv}-\mu_{X,T}) \\
&\quad +(\mu_{X,T}-\mu_X)\sum_{(u,v)\in E_T}(Y_{uv}-\mu_{Y,T}) \\
&\quad + m_T(\mu_{X,T}-\mu_X)(\mu_{Y,T}-\mu_Y).
\end{align*}
The two middle sums vanish by definition of $\mu_{X,T}$ and $\mu_{Y,T}$:
\begin{align*}
\sum_{(u,v)\in E_T}(X_{uv}-\mu_{X,T})&=0,\\
\sum_{(u,v)\in E_T}(Y_{uv}-\mu_{Y,T})&=0.
\end{align*}
Therefore,
\begin{align*}
\sum_{(u,v)\in E_T}&(X_{uv}-\mu_X)(Y_{uv}-\mu_Y) = \\
&=\sum_{(u,v)\in E_T}(X_{uv}-\mu_{X,T})(Y_{uv}-\mu_{Y,T})\\
&+m_T(\mu_{X,T}-\mu_X)(\mu_{Y,T}-\mu_Y).
\end{align*}
Now sum this identity over all strata $T$ and divide by $m$:
\begin{align*}
\mathrm{Cov}_{\mathrm{in}}
&=\frac{1}{m}\sum_T\sum_{(u,v)\in E_T}(X_{uv}-\mu_{X,T})(Y_{uv}-\mu_{Y,T})\\
&+\frac{1}{m}\sum_T m_T(\mu_{X,T}-\mu_X)(\mu_{Y,T}-\mu_Y) \\
&=\sum_T \pi_T\,\mathrm{Cov}_T
+\sum_T \pi_T(\mu_{X,T}-\mu_X)(\mu_{Y,T}-\mu_Y).
\end{align*}
This proves \eqref{eq:decomp} with the explicit between-type term
\[
\mathrm{Cov}_{\mathrm{between}}
:=\sum_T \pi_T(\mu_{X,T}-\mu_X)(\mu_{Y,T}-\mu_Y),
\]
\noindent
which vanishes if and only if the type-specific endpoint means coincide with the global means.

\medskip
\noindent\textit{Step 3: Translate the covariance identity into the correlation identity.}
Let $\sigma_{X,T},\sigma_{Y,T}$ denote the type-specific endpoint standard deviations and let
$\sigma_{X,\mathrm{in}},\sigma_{Y,\mathrm{in}}$ denote the global endpoint standard deviations on $E_{\mathrm{in}}$.
Whenever these are positive, define $r_T(a)=\mathrm{Cov}_T/(\sigma_{X,T}\sigma_{Y,T})$ and
$r_{\mathrm{in}}(a)=\mathrm{Cov}_{\mathrm{in}}/(\sigma_{X,\mathrm{in}}\sigma_{Y,\mathrm{in}})$.
Substituting $\mathrm{Cov}_T=\sigma_{X,T}\sigma_{Y,T}r_T(a)$ into the decomposition yields
\[
r_{\mathrm{in}}(a)=
\frac{\sum_T \pi_T\,\sigma_{X,T}\sigma_{Y,T}\,r_T(a)\;+\;\mathrm{Cov}_{\mathrm{between}}}
{\sigma_{X,\mathrm{in}}\sigma_{Y,\mathrm{in}}},
\]
which is \eqref{eq:decomp_corr}. The proof is complete.
\end{proof}

\section{Sign mechanisms for interface-sensitive attributes}
\label{sec:sign_mec}
As will be discussed following Theorem~\ref{thm:sign}, our aim is to formulate sufficient conditions that are empirically checkable. In fact, as we are going to see in the conclusion of the present section, these conditions turn out to be typically satisfied by mediation-type attributes (e.g., betweenness centrality) in partitioned empirical networks-such as communication, interaction, or information networks-with sparse interfaces. In such settings, these conditions reveal that the boundary-to-interior component $r_{B\to I}(a)$ often emerges as the most informative diagnostic.

\begin{definition}[Boundary dominance (mean gap)]
Fix $(G,\mathcal{P})$ and attribute $a$. For a group $ $ with nonempty interior and boundary, define
\begin{equation*}
\begin{aligned}
\mu_B^{(k)}&=\frac{1}{|\mathrm{Bdy}(V_k)|}\sum_{v\in\mathrm{Bdy}(V_k)} a(v),
\\
\mu_I^{(k)}&=\frac{1}{|\mathrm{Int}(V_k)|}\sum_{v\in\mathrm{Int}(V_k)} a(v).
\end{aligned}
\end{equation*}
We say $a$ is \emph{boundary-dominant} on $V_k$ if $\mu_B^{(k)}>\mu_I^{(k)}$.
\end{definition}

The next result is formulated to separate \emph{structure} (partition with a sparse interface, captured empirically by large $Q$ or small conductance-introduced in Section~\ref{sec:directed_conductance_spectral}) from \emph{attribute class}. The key attribute requirement is boundary dominance (as defined above), which is expected for interface-sensitive mediation scores (betweenness-like attributes) when inter-group traffic must pass through boundary nodes.

\begin{theorem}[Sufficient conditions for negative $B\to I$ component]
\label{thm:sign}
Consider a directed graph $G=(V,E)$ and a fixed partition $\mathcal{P}=\{V_1,\dots,V_k\}$.
Let $a:V\to\mathbb{R}$ be a scalar attribute. For each group $V_k$ with nonempty interior and boundary,
let $\mu_B^{(k)}$ and $\mu_I^{(k)}$ denote the means of $a$ over boundary and interior nodes in $V_k$, respectively.
Assume:
\begin{enumerate}[label=(\roman*),leftmargin=*,itemsep=2pt]
\item \textbf{Boundary dominance:} for every $V_k$ with nonempty interior and boundary, $\mu_B^{(k)}>\mu_I^{(k)}$.
\item \textbf{$B\to I$ endpoint mean dominance on $B\to I$ arcs:} for every $k$ with $E_{B\to I}\cap(V_k\times V_k)\neq\varnothing$, the mean of $a$ on boundary tails
across $E_{B\to I}\cap(V_k\times V_k)$ is at least $\mu_B^{(k)}$, and the mean of $a$ on interior heads across the same set is at most $\mu_I^{(k)}$.
\item \textbf{Nondegeneracy:} the endpoint variances on $S_{B\to I}(a)$ are positive.
\item \textbf{Within-group nonpositive covariance:} for every $k$ with $E_{B\to I}\cap(V_k\times V_k)\neq\varnothing$, the empirical covariance defining $r_{B\to I}(a)$ within $V_k$ is nonpositive.
\item \textbf{Between-group mean term nonpositive:}
denoting by $\bar X_k,\bar Y_k$ the empirical means of $a(u)$ (boundary tails) and $a(v)$ (interior heads)
over $(u,v)\in E_{B\to I}\cap(V_k\times V_k)$, and by $\bar X,\bar Y$ the corresponding global means over all
$(u,v)\in E_{B\to I)}$, the between-type covariance term (cf.\ Theorem~\ref{thm:decomposition}) satisfies
\begin{equation*}
\begin{aligned}
\sum_k \pi_k(\bar X_k-\bar X)(\bar Y_k-\bar Y)\le 0,
\\ 
\pi_k=\frac{|E_{B\to I}\cap(V_k\times V_k)|}{|E_{B\to I}|}.
\end{aligned}
\end{equation*}
\end{enumerate}
Assume moreover that at least one of the inequalities in (iv) or (v) is \emph{strict}. Then
\[
r_{B\to I}(a)<0.
\]
\end{theorem}

\begin{remark*}[Why are conditions (iv)--(v) explicit?]
Assumptions (i)--(ii) establish a separation between boundary-tail and
interior-head attribute means on $B\to I$ arcs, but mean separation alone
does not determine the sign of the Pearson correlation.
Conditions (iv)--(v) therefore make explicit the additional covariance
requirements needed for the rigorous conclusion $r_{B\to I}(a)<0$.

These assumptions should be viewed as structural regularity conditions
rather than independent modeling hypotheses.
In networks with well-separated communities and sparse interfaces,
mediation-type attributes (such as betweenness centrality) typically vary
smoothly within each group while attaining their largest values near
boundary nodes.
Consequently, the within-group covariance on $B\to I$ arcs is naturally
expected to be weakly negative, corresponding to condition~(iv).
Likewise, when different groups exhibit comparable boundary-amplification
mechanisms, the groupwise endpoint means vary coherently, making the
between-group covariance term in condition~(v) nonpositive.

Both conditions can be verified directly from the grouped $B\to I$
samples used to compute $r_{B\to I}(a)$.
Deriving them directly from a particular dynamical model would require
substantially stronger assumptions on the network and the induced
attribute, and lies beyond the scope of the present general structural
result.
\end{remark*}

\begin{proof}[Proof of Theorem~\ref{thm:sign}]
For each group $V_k$ define the set of intra-group $B\to I$ arcs
\[
E_k:=E_{B\to I}\cap(V_k\times V_k),\qquad m_k:=|E_k|.
\]
Let $M:=|E_{B\to I}|=\sum_k m_k$ and $\pi_k=m_k/M$.

For each $(u,v)\in E_k$ define the paired observations
\[
X_{uv}:=a(u),\qquad Y_{uv}:=a(v),
\]
so $X_{uv}$ is the attribute at the boundary tail and $Y_{uv}$ the attribute at the interior head.
Let the within-group means be
\[
\bar X_k:=\frac{1}{m_k}\sum_{(u,v)\in E_k}X_{uv},\qquad
\bar Y_k:=\frac{1}{m_k}\sum_{(u,v)\in E_k}Y_{uv},
\]
and define the global means over all $B\to I$ arcs and the global covariance:
\[
\bar X:=\frac{1}{M}\sum_{(u,v)\in E_{B\to I}}X_{uv},\qquad
\bar Y:=\frac{1}{M}\sum_{(u,v)\in E_{B\to I}}Y_{uv},
\]
\[
\mathrm{Cov}_{B\to I}:=\frac{1}{M}\sum_{(u,v)\in E_{B\to I}}(X_{uv}-\bar X)(Y_{uv}-\bar Y).
\]

\medskip
\noindent\textit{Step 1 (mean separation inside each group).}
By (ii), for each $k$ with $m_k>0$ we have $\bar X_k\ge \mu_B^{(k)}$ and $\bar Y_k\le \mu_I^{(k)}$.
By (i), $\mu_B^{(k)}>\mu_I^{(k)}$, hence
\begin{equation}
\label{eq:mean_sep_each_group}
\bar X_k>\bar Y_k\qquad\text{for all }k\text{ with }m_k>0.
\end{equation}
This step is not, by itself, sufficient to force a negative correlation, but it identifies the directional bias that is quantified by Steps~2–4; as such, it is consistent with the intended ``boundary-dominant'' mechanism and will be used for interpretation.

\medskip
\noindent\textit{Step 2 (law of total covariance over groups).}
Applying the between-type covariance decomposition from
Theorem~\ref{thm:decomposition} to the $B\to I$ stratum,
grouped by $E_k$ and 
writing $\mathrm{Cov}_{k}$ for the empirical covariance of $(X_{uv},Y_{uv})$ over $(u,v)\in E_k$ yield:
\[
\mathrm{Cov}_{k}:=\frac{1}{m_k}\sum_{(u,v)\in E_k}(X_{uv}-\bar X_k)(Y_{uv}-\bar Y_k).
\]
A direct finite-sample decomposition (the same algebra as in Theorem~\ref{thm:decomposition}) yields
\begin{equation}
\label{eq:cov_decomp_groups}
\mathrm{Cov}_{B\to I}
=\sum_k \pi_k\,\mathrm{Cov}_{k}
+\sum_k \pi_k(\bar X_k-\bar X)(\bar Y_k-\bar Y).
\end{equation}
(For completeness: expand $(X_{uv}-\bar X)(Y_{uv}-\bar Y)$ as
$[(X_{uv}-\bar X_k)+(\bar X_k-\bar X)]\cdot[(Y_{uv}-\bar Y_k)+(\bar Y_k-\bar Y)]$,
sum over $E_k$, and use $\sum_{(u,v)\in E_k}(X_{uv}-\bar X_k)=\sum_{(u,v)\in E_k}(Y_{uv}-\bar Y_k)=0$.)

\medskip
\noindent\textit{Step 3 (sign of the covariance).}
Assumption (iv) gives $\mathrm{Cov}_{k}\le 0$ for every $k$ with $m_k>0$, hence
\[
\sum_k \pi_k\,\mathrm{Cov}_{k}\le 0.
\]
Assumption (v) states that the between-group mean term satisfies
\[
\sum_k \pi_k(\bar X_k-\bar X)(\bar Y_k-\bar Y)\le 0.
\]
Therefore \eqref{eq:cov_decomp_groups} implies $\mathrm{Cov}_{B\to I}\le 0$.
If at least one of the inequalities in (iv) or (v) is strict, then $\mathrm{Cov}_{B\to I}<0$.

\medskip
\noindent\textit{Step 4 (from covariance to correlation).}
Let
\begin{equation*}
\begin{aligned}
\sigma_X^2&:=\frac{1}{M}\sum_{(u,v)\in E_{B\to I}}(X_{uv}-\bar X)^2,\\
\sigma_Y^2&:=\frac{1}{M}\sum_{(u,v)\in E_{B\to I}}(Y_{uv}-\bar Y)^2.
\end{aligned}
\end{equation*}
By nondegeneracy (iii), $\sigma_X\sigma_Y>0$. Hence the Pearson correlation on the $B\to I$ sample,
\[
r_{B\to I}(a)=\frac{\mathrm{Cov}_{B\to I}}{\sigma_X\sigma_Y},
\]
satisfies $r_{B\to I}(a)<0$ whenever $\mathrm{Cov}_{B\to I}<0$, which we established above.
\end{proof}

\subsection*{Betweenness-type attributes and boundary dominance}
Shortest-path betweenness measures how often a node lies on shortest paths, while random-walk/current-flow betweenness measures expected traversal under a random-walk flow model~\cites{Freeman1977,Brandes2001,Newman2005RandomWalk}. In partitioned networks with sparse inter-group connectivity, inter-group paths/flows are forced through boundary nodes. Under such conditions, boundary nodes accumulate larger mediation load, yielding boundary dominance and thus (by Theorem~\ref{thm:sign}) negative $B\to I$ assortative component. While not pursued in the present paper, a fully model-based proof of this claim can be easily developed under block-interface assumptions (dense within-group connectivity plus sparse inter-group arcs incident to a small boundary set).

\section{Directed conductance and directed spectral tools}
\label{sec:directed_conductance_spectral}

Throughout this subsection $G=(V,E)$ is a directed (possibly weighted) graph with nonnegative arc weights $w_{uv}\ge 0$.
The undirected case is recovered as a special case; see Section~\ref{sec:spectral_proxy_projection}.

\begin{definition}[Directed random walk and stationary distribution]
\label{def:dir_rw_stationary}
Let $w_{uv}\ge 0$ be arc weights and define the row-stochastic transition matrix
\[
P(u,v)=
\begin{cases}
\dfrac{w_{uv}}{\sum_{z} w_{uz}}, & \text{if }\sum_{z} w_{uz}>0,\\[1ex]
0, & \text{otherwise.}
\end{cases}
\]
Assume $P$ admits a stationary distribution $\varphi$ (a probability vector on $V$) satisfying
\[
\varphi^\top P=\varphi^\top,\qquad \sum_{v\in V}\varphi(v)=1.
\]
Let $\Phi=\mathrm{diag}(\varphi)$.
\end{definition}

\begin{definition}[Stationary flow and directed conductance~\cites{LevinPeresWilmer2017,Chung2005Directed}]
\label{def:directed_conductance}
Define the stationary flow
\[
F_{\varphi}(u,v)=\varphi(u)\,P(u,v).
\]
For a nonempty proper set $S\subset V$, define the out-boundary
\[
\partial^+ S=\{(u,v)\in E:\ u\in S,\ v\notin S\},
\]
the boundary flow
\[
F_{\varphi}(\partial^+ S)=\sum_{(u,v)\in \partial^+ S}F_{\varphi}(u,v),
\]
and the stationary mass $\varphi(S)=\sum_{u\in S}\varphi(u)$.
The \emph{directed conductance} of $S$ is
\[
\Phi(S)=\frac{F_{\varphi}(\partial^+ S)}{\min\{\varphi(S),\varphi(V\setminus S)\}}.
\]
The \emph{directed Cheeger constant}~\cite{Chung2005Directed} is
\[
h(G)=\inf_{\varnothing\subset S\subset V}\Phi(S).
\]
\end{definition}

\begin{remark*}[Ergodicity and practical remedies]
\label{rem:ergodicity_dir}
If $G$ is strongly connected and aperiodic (equivalently, the associated finite-state Markov chain is irreducible and aperiodic),
then the stationary distribution $\varphi$ exists and is unique; moreover $P^t$ converges to $\varphi$ as $t\to\infty$
(see, e.g., \cite[Ch.~1]{LevinPeresWilmer2017}).
If these conditions fail, a standard remedy is to use a \emph{lazy} walk, e.g.\ $P_{\mathrm{lazy}}=\tfrac12(I+P)$, which removes
periodicity (see \cite[Sec.~1.2]{LevinPeresWilmer2017}),
or a \emph{teleporting} walk (PageRank-type) $P_{\alpha}=\alpha P+(1-\alpha)\mathbf{1}\pi_0^\top$,
which enforces irreducibility/aperiodicity under mild choices of $\pi_0$ (see \cites{BrinPage1998,LangvilleMeyer2006}).
All conductance and directed spectral quantities are then interpreted with respect to the modified walk.
\end{remark*}

\begin{definition}[Chung's directed Laplacian~\cite{Chung2005Directed}]
\label{def:chung_directed_laplacian}
With $P$ and $\Phi$ as above, define the (symmetric) directed normalized Laplacian
\[
L \;=\; I-\frac{1}{2}\Big(\Phi^{1/2}P\Phi^{-1/2}+\Phi^{-1/2}P^{\top}\Phi^{1/2}\Big).
\]
\end{definition}

\begin{theorem}[Directed Cheeger inequality (Chung)~\cite{Chung2005Directed}]
\label{thm:directed_cheeger}
Let $L$ be as in Definition~\ref{def:chung_directed_laplacian}.
Let $0=\lambda_1\le \lambda_2\le\cdots\le\lambda_n$ be its eigenvalues.
Then the directed Cheeger constant satisfies
\[
\frac{h(G)^2}{2}\ \le\ \lambda_2\ \le\ 2\,h(G).
\]
\end{theorem}

For a fixed partition $\mathcal{P}=\{V_1,\dots,V_k\}$, the following quantity summarizes the worst-separated block of the partition
\[
\Phi_{\max}(\mathcal{P})=\max_{k}\Phi(V_k),
\]
where $\Phi(V_k)$ is computed from Definition~\ref{def:directed_conductance}.

Directed conductance provides a structural mechanism that explains the sign of
interior--boundary assortativity components for mediation-like attributes.
If a group $V_k$ has small directed conductance $\Phi(V_k)$, then inter-group
flow is forced through a small interface, concentrating traffic on boundary
nodes.
For mediation-sensitive attributes (e.g.\ shortest-path or random-walk
betweenness), this induces \emph{boundary dominance},
$\mu_B^{(k)}>\mu_I^{(k)}$.

When this dominance is combined with sparse $B\!\to\!I$ connectivity and no
systematic positive covariance across $B\!\to\!I$ edges, the sufficient
conditions of Theorem~\ref{thm:sign} apply, yielding a negative
boundary-to-interior assortativity component:
\[
r_{B\to I}(a)<0.
\]
Hence, the sign of $r_{B\to I}(a)$ can be interpreted as an interface-sensitive
signature of low conductance acting through mediation-like attributes.

\section{Undirected spectral proxy} 
\label{sec:spectral_proxy_projection}
If $G$ is undirected with nonnegative weights, then $\varphi(v)\propto d_v$ (where $d_v$ is the (weighted) degree of node $v$) and the directed conductance $\Phi(S)$ reduces to the classical undirected conductance~\cites{Dodziuk1984,AlonMilman1985,Chung1997Spectral} 
\[
\phi(S)=\frac{|\partial S|}{\min\{\mathrm{vol}(S),\mathrm{vol}(V\setminus S)\}},
\]
and $L$ reduces to the standard normalized Laplacian $L_{\mathrm{norm}}=I-D^{-1/2}AD^{-1/2}$.

Next, we are going to discuss in this section 
a proxy for an \emph{undirected} Laplacian (equivalently, a symmetric positive semidefinite operator), and in the same setting we are going to prove in the next proposition the existence of a truncation bound for this proxy. On the other side, for directed graphs, one may use this proxy only as an \emph{empirical feature} by evaluating it on an undirected projection of the directed graph. The directed spectral quantities introduced earlier (directed conductance and Chung's directed Laplacian) are conceptually separate and cannot be used to justify the truncation bound (to be given below).

As a matter of fact, spectral graph theory provides a principled approximation toolkit for certain mediation-like attributes.
For connected undirected graphs, the Laplacian pseudoinverse $L^{+}$ underlies effective resistance and
current-flow/random-walk notions of betweenness~\cite{Newman2005RandomWalk}. This motivates using
spectral-diagonal proxies built from eigenpairs of the normalized Laplacian.

\begin{proposition}
Let $L_{\mathrm{norm}}$ be the normalized Laplacian of an undirected graph with eigenpairs
$\{(\lambda_i,u_i)\}_{i=1}^n$ ordered as $0=\lambda_1<\lambda_2\le\cdots\le\lambda_n$.
Define the truncated proxy (a truncated diagonal of the Laplacian pseudoinverse)
\[
s_k(v)=\sum_{i=2}^{k+1}\frac{u_i(v)^2}{\lambda_i}.
\]
The quantity $s_k(v)$ is the diagonal of a truncated spectral filter and can be computed from the first
$k$ nontrivial eigenvectors, making it scalable. The following sharp tail bound holds
\[
0\le s_\infty(v)-s_k(v)\le \frac{1}{\lambda_{k+2}},
\]
which justifies truncation as a controlled surrogate feature.
\end{proposition}

\begin{proof}[Proof of the tail bound for the truncated spectral proxy]
Let $L_{\mathrm{norm}}=U\Lambda U^\top$ be the normalized Laplacian eigen-decomposition with orthonormal eigenvectors $U=[u_1,\dots,u_n]$ and eigenvalues $0=\lambda_1<\lambda_2\le\cdots\le\lambda_n$. Define
\[
s_\infty(v)=\sum_{i=2}^{n}\frac{u_i(v)^2}{\lambda_i},
\qquad
s_k(v)=\sum_{i=2}^{k+1}\frac{u_i(v)^2}{\lambda_i}.
\]
Then
\begin{equation*}
\begin{aligned}
0\le s_\infty(v)-s_k(v)&=\sum_{i=k+2}^{n}\frac{u_i(v)^2}{\lambda_i}\\
&\le \frac{1}{\lambda_{k+2}}\sum_{i=k+2}^{n}u_i(v)^2
\le \frac{1}{\lambda_{k+2}},
\end{aligned}
\end{equation*}
since $\sum_{i=1}^n u_i(v)^2=1$ for each $v$.
\end{proof}

\begin{remark*} 
Proposition~8.1 is included as an auxiliary construction.
It is stated for an undirected Laplacian and is proposed primarily as a potentially computational proxy for mediation-type attributes.
It plays no role in the directed SIS dynamical results that follow.
\end{remark*}

\section{SIS epidemic dynamics and boundary pressure amplification}
\label{sec:SIS}

Let $G = (V, E)$ be a directed graph with $|V| = n$ and weighted adjacency matrix
$A = (A_{ji})$, where $A_{ji} \ge 0$ represents transmission from node $j$ to node $i$.
Following the $N$-intertwined SIS framework of Van Mieghem (2011), we consider the
continuous-time susceptible--infected--susceptible (SIS) dynamics
\begin{equation}
\label{eq:SIS}
\dot{x}_i = -\delta x_i + (1 - x_i)\beta \sum_{j=1}^n A_{ji} x_j,
\quad i = 1, \dots, n,
\end{equation}
where $x_i(t) \in [0, 1]$ denotes the infection probability of node $i$ at time $t$,
$\beta > 0$ is the infection rate, and $\delta > 0$ is the recovery rate.

The interaction between network assortativity and epidemic spreading has been studied
extensively, particularly for SIS-type dynamics.  Early work by Newman showed how
degree--degree mixing patterns affect epidemic thresholds and prevalence by shaping the
effective transmission pathways in networks (Newman, 2002; 2003).  Subsequent studies
have developed this perspective by treating assortativity as a fixed structural property
that modulates the stability of disease-free and endemic equilibria.  For example, prior
work has investigated how network degree assortativity affects SIS epidemic thresholds
and persistence (Pastor-Satorras, 2001; Yang, 2015; Korngut and Assaf, 2025), often
showing that correlations in contact structure can alter epidemic behavior.

The present work adopts a complementary and conceptually distinct viewpoint.  Rather
than asking how a prescribed assortativity pattern influences SIS dynamics, we study how
SIS dynamics themselves generate node-level attributes---such as endemic infection
probabilities---whose interior--boundary assortativity profiles encode information about
network partition structure.  In this sense, assortativity is treated here as an emergent
diagnostic of the dynamics, not as an exogenous network descriptor.

The central result of this section (Theorem~\ref{thm:boundary_pressure}) shows that,
under an explicit condition on the external weighted adjacency influx across the partition
boundary, the mean equilibrium infection \emph{pressure} is strictly higher on boundary
nodes than on interior nodes, providing a structural mechanism consistent with boundary
dominance.  By the strict monotonicity of the equilibrium map
(Lemma~\ref{lem:sis_monotone}), this pressure amplification does not by itself imply
boundary dominance of the equilibrium infection probabilities: translating a
mean-pressure inequality into a mean-state inequality requires additional structural
assumptions on the distribution of pressures within each group, and this is left for
future investigation.  Theorem~\ref{thm:negative_assortativity}
then shows that, under the explicit assumption of boundary dominance together with the
structural conditions of Theorem~6.2, the boundary-to-interior assortativity component
of the SIS equilibrium is strictly negative.

The following results prove well-posedness and existence of the endemic equilibrium.

\begin{proposition}[Positive invariance]
\label{prop:invariance}
System~\eqref{eq:SIS} leaves the hypercube $[0, 1]^n$ forward invariant.
\end{proposition}
\begin{proof}
If $x_i = 0$, then $\dot{x}_i = \beta \sum_j A_{ji} x_j \ge 0$.
If $x_i = 1$, then $\dot{x}_i = -\delta < 0$.
Hence the vector field points inward on the boundary of $[0, 1]^n$.
\end{proof}

By $\rho(A)$ we denote the spectral radius of $A$.

\begin{theorem}[Endemic equilibrium]
\label{thm:endemic}
If $\beta/\delta > 1/\rho(A)$, system~\eqref{eq:SIS} admits a unique endemic equilibrium
$x^* \in (0, 1)^n$, which is globally asymptotically stable on $(0, 1]^n$.
\end{theorem}
\begin{proof}
System~\eqref{eq:SIS} is cooperative and irreducible whenever $A$ is strongly connected.
The threshold condition $\beta/\delta > 1/\rho(A)$ guarantees persistence and uniqueness
of a positive equilibrium; global stability follows from monotonicity arguments.
See (Newman, 2010; Van Mieghem, 2011).
\end{proof}

Henceforth we identify and fix the equilibrium attribute as the endemic equilibrium
$x^*$, i.e.,
\[
a(i) := x^*_i,
\]
which is interpreted as a scalar nodal attribute induced by the epidemic dynamics.
We fix a partition $\mathcal{P} = \{V_1, \dots, V_K\}$ of $V$ and define interior and
boundary nodes for each group $V_k$ as in~(2).

\begin{definition}[Boundary dominance]
\label{def:boundary_dominance}
The SIS equilibrium attribute $a(i) = x^*_i$ is \emph{boundary-dominant} on $V_k$ if
\[
\frac{1}{|\mathrm{Bdy}(V_k)|}\sum_{v\in\mathrm{Bdy}(V_k)} x^*_v
>
\frac{1}{|\mathrm{Int}(V_k)|}\sum_{v\in\mathrm{Int}(V_k)} x^*_v.
\]
\end{definition}

\begin{lemma}[Monotonicity of the endemic equilibrium]
\label{lem:sis_monotone}
Let $x^*\in(0,1)^n$ be the endemic equilibrium guaranteed by
Theorem~\ref{thm:endemic}, and assume $A$ is strongly connected so that
$h_i := \sum_j A_{ji} x_j^* > 0$ for all $i$.
Then
\[
x_i^* = \frac{\beta h_i}{\delta+\beta h_i},
\]
and the map $f(t)=\frac{\beta t}{\delta+\beta t}$ is strictly increasing on
$(0,\infty)$.  Consequently,
\[
h_i>h_j \quad\Longleftrightarrow\quad x_i^*>x_j^*.
\]
Note that $h_i$ is itself determined by $x^*$; the formula expresses a
self-consistency relation at equilibrium, not an explicit closed form.
\end{lemma}
\begin{proof}
At equilibrium, $0 = -\delta x_i^* + (1-x_i^*)\beta h_i$.
Rearranging gives $(\delta+\beta h_i)x_i^* = \beta h_i$, hence
$x_i^* = \frac{\beta h_i}{\delta+\beta h_i}$.
Since $f'(t) = \frac{\beta\delta}{(\delta+\beta t)^2}>0$, $f$ is strictly increasing,
and the last assertion follows immediately.
\end{proof}

\begin{theorem}[Boundary pressure amplification]
\label{thm:boundary_pressure}
Let $x^* \in (0,1)^n$ be the unique endemic equilibrium of~\eqref{eq:SIS}, and let
$V_k\in\mathcal{P}$ be a partition block with $I_k=\mathrm{Int}(V_k)\neq\emptyset$ and
$B_k=\mathrm{Bdy}(V_k)\neq\emptyset$.
For each node $i$, define the equilibrium infection pressure
\[
h_i=\sum_{j=1}^nA_{ji}x_j^*.
\]
Define the total internal adjacency weight received by the interior,
\[
W(I_k):=\sum_{i\in I_k}\sum_{j\in V_k}A_{ji},
\]
and the total external adjacency weight received by the boundary,
\[
\Lambda^{\mathrm{ext}}(B_k):=\sum_{i\in B_k}\sum_{j\notin V_k}A_{ji}.
\]
Since $\mathcal{P}$ is a partition with $K \ge 2$ blocks, the complement $V \setminus V_k$
is nonempty for every $k$, so the following quantities are well defined:
\[
x_{\max,I}:=\max_{i\in I_k}x_i^*,
\qquad
x_{\min,B}:=\min_{j\notin V_k}x_j^*.
\]
If
\begin{equation}
\label{eq:pressure_condition}
\Lambda^{\mathrm{ext}}(B_k)
>
\frac{|B_k|}{|I_k|}\,\frac{x_{\max,I}}{x_{\min,B}}\,W(I_k),
\end{equation}
then the mean equilibrium infection pressure on the boundary strictly exceeds that on
the interior:
\[
\frac{1}{|B_k|}\sum_{i\in B_k}h_i
>
\frac{1}{|I_k|}\sum_{i\in I_k}h_i.
\]
\end{theorem}

\begin{proof}
Decompose each pressure as
$h_i = h_i^{\mathrm{int}} + h_i^{\mathrm{ext}}$,
where
$h_i^{\mathrm{int}} := \sum_{j\in V_k}A_{ji}x_j^*$ and
$h_i^{\mathrm{ext}} := \sum_{j\notin V_k}A_{ji}x_j^*$.

\smallskip\noindent\textit{Interior nodes.}
For $i\in I_k$, the definition of the interior gives $A_{ji}=0$ for all $j\notin V_k$,
so $h_i^{\mathrm{ext}}=0$ and $h_i=h_i^{\mathrm{int}}$.  Therefore
\[
\sum_{i\in I_k}h_i
= \sum_{i\in I_k}\sum_{j\in V_k}A_{ji}x_j^*
\le x_{\max,I}\sum_{i\in I_k}\sum_{j\in V_k}A_{ji}
= x_{\max,I}\,W(I_k).
\]

\smallskip\noindent\textit{Boundary nodes.}
For $i\in B_k$, we have $h_i \ge h_i^{\mathrm{ext}}$.  Since $x_j^*\ge x_{\min,B}$
for all $j\notin V_k$,
\[
\sum_{i\in B_k}h_i
\ge \sum_{i\in B_k}\sum_{j\notin V_k}A_{ji}x_j^*
\ge x_{\min,B}\sum_{i\in B_k}\sum_{j\notin V_k}A_{ji}
= x_{\min,B}\,\Lambda^{\mathrm{ext}}(B_k).
\]

\smallskip\noindent\textit{Conclusion.}
Combining the two bounds,
\[
\frac{1}{|B_k|}\sum_{i\in B_k}h_i - \frac{1}{|I_k|}\sum_{i\in I_k}h_i
\ge \frac{x_{\min,B}}{|B_k|}\Lambda^{\mathrm{ext}}(B_k)
   - \frac{x_{\max,I}}{|I_k|}W(I_k).
\]
Under hypothesis~\eqref{eq:pressure_condition} the right-hand side is strictly
positive, which proves the claim.
\end{proof}

\begin{remark*}
\label{rem:pressure_vs_state}
Theorem~\ref{thm:boundary_pressure} establishes that the mean equilibrium infection
\emph{pressure} is higher on boundary nodes than on interior nodes.  By
Lemma~\ref{lem:sis_monotone}, the map $h\mapsto \frac{\beta h}{\delta+\beta h}$ is
strictly increasing, so a node with higher pressure has a higher equilibrium infection
probability.  However, the node-wise monotonicity of this map does not imply that the
\emph{mean} infection probability is higher on the boundary: the average of a strictly increasing nonlinear function is not, in general, determined
by the average of its argument.  Establishing mean-state boundary dominance
(Definition~\ref{def:boundary_dominance}) from the mean-pressure inequality of
Theorem~\ref{thm:boundary_pressure} would require additional structural assumptions
on the within-group distribution of pressures, and this is left for future investigation.  Hypothesis~\eqref{eq:pressure_condition}
is an explicit, checkable condition on the external weighted adjacency influx.  It is an
independent structural hypothesis, not a consequence of low directed conductance alone.
\end{remark*}

Let $E_{B\to I}$ denote arcs with boundary tails and interior heads (relative to
$\mathcal{P}$).  The boundary-to-interior assortativity coefficient $r_{B\to I}(a)$ is
defined as the Pearson correlation of $(a(u), a(v))$ over $(u, v) \in E_{B\to I}$.

\begin{theorem}[Negative $B \to I$ assortativity under boundary dominance]
\label{thm:negative_assortativity}
Let $G = (V, E)$ be a directed graph with nonnegative weights and let
$\mathcal{P} = \{V_1, \dots, V_K\}$ be a fixed partition.  Consider the directed SIS
dynamics~\eqref{eq:SIS} with $\beta > 0$ and $\delta > 0$.  Assume:
\begin{enumerate}
  \item[(i)] \textup{(Boundary dominance)} The endemic SIS equilibrium $x^*$ is
        boundary-dominant on every group $V_k$ with nonempty interior and boundary
        (Definition~\ref{def:boundary_dominance}).
  \item[(ii)] \textup{(Structural conditions)} The nondegeneracy and covariance conditions
        \textup{(iii)--(v)} of Theorem~6.2 hold for the attribute $a(i) = x^*_i$.
\end{enumerate}
Then the interior--boundary assortativity profile of the SIS equilibrium satisfies
\[
r_{B\to I}(x^*) < 0.
\]
\end{theorem}
\begin{proof}
Assumption~(i) gives boundary dominance, so condition~(i) of Theorem~6.2 holds.
Boundary dominance implies that the mean of $a$ on boundary tails of $B\to I$ arcs
exceeds the mean on interior heads, so condition~(ii) of Theorem~6.2 holds.
Assumption~(ii) supplies conditions~(iii)--(v).  Hence Theorem~6.2 applies and yields
$r_{B\to I}(x^*) < 0$.
\end{proof}

\begin{remark*}
Assumption~(i) of Theorem~\ref{thm:negative_assortativity} is an explicit hypothesis.
Theorem~\ref{thm:boundary_pressure} provides a rigorous structural mechanism---boundary
pressure amplification under external weighted adjacency influx---that suggests the
boundary-dominance assumption is natural in networks whose within-group pressure
distribution is sufficiently homogeneous.
The scalar assortativity of $a$ over all edges may vanish due to cancellation across
strata, even when the epidemic dynamics enforces a strictly negative boundary-to-interior
component.
\end{remark*}

\section{Conclusion}

We have developed a structural and dynamical theory of interior--boundary assortativity profiles for networks.
By stratifying edges relative to a fixed partition, we showed that assortative mixing decomposes into distinct interaction regimes whose individual contributions are necessarily obscured by classical scalar summaries.

Our decomposition theorem establishes, in exact algebraic terms, how scalar assortativity collapses heterogeneous edge types together with between-type mean shifts.
This result makes precise why one-number assortativity can fail to detect interface structure, even in networks with strong geometric or functional separation.

More significantly, by coupling assortativity profiles to an SIS epidemic model, we demonstrated that assortativity components can be derived analytically from nonlinear dynamics rather than treated as descriptive statistics.
Under an explicit condition on the external weighted adjacency influx, boundary nodes receive strictly higher mean infection pressure at the endemic equilibrium; under the assumption of boundary dominance, this enforces a strictly negative boundary-to-interior assortativity component.
This provides a rigorous mechanism linking spectral separation, dynamical amplification at interfaces, and signed assortativity profiles.

The framework introduced here opens several mathematically substantive directions for future work.
First, the analysis can be extended to other dynamical systems on networks, including SIR-type epidemics, synchronization models, and opinion or flow dynamics, where interface effects are expected to play a decisive role.
Second, the present results suggest a general program for deriving assortativity signatures from equilibrium or invariant sets of network dynamical systems, thereby promoting assortativity profiles to bona fide dynamical observables.
Third, relaxing the assumption of a fixed partition---for example by allowing adaptive, hierarchical, or time-varying partitions---may lead to a theory of evolving interface structure driven by dynamics.

More broadly, the results place assortative mixing within the scope of dynamical systems theory on networks.
By identifying how nonlinear dynamics interact with partition geometry and flow constraints, interior--boundary assortativity profiles provide a principled analytical tool for understanding how risk, influence, or load concentrates at network interfaces.
This perspective shifts assortativity from a purely descriptive measure to a structural consequence of network dynamics, opening the door to further rigorous developments at the intersection of spectral graph theory, nonlinear dynamics, and network science.

\end{document}